\documentclass[letterpaper,twocolumn,10pt]{article}
\usepackage{usenix,epsfig}

\usepackage[latin1]{inputenc}
\usepackage{amsmath}
\usepackage{amsfonts}
\usepackage{amssymb}
\usepackage{amsthm}
\usepackage{color}
\usepackage{subfig}
\usepackage{url}
\usepackage{ulem}
\usepackage{comment}

\usepackage{csquotes}
\usepackage{balance}
\usepackage[nolist,nohyperlinks]{acronym}
\usepackage{array}

\usepackage[ruled,linesnumbered]{algorithm2e}

\newtheorem{proposition}{Proposition}
\newcommand*\diff{\mathop{}\!\mathrm{d}}
\DeclareMathOperator{\E}{\mathbb{E}}

\newcommand{\ie}{\textit{i.e.}\xspace}

\begin{document}
\normalem
\date{}

\title{\Large \bf Elastic Provisioning of Cloud Caches:\\a Cost-aware TTL Approach}

\author{
{\rm Damiano Carra}\\
University of Verona
\and
{\rm Giovanni Neglia}\\
Universit\'e C\^ote d'Azur, Inria
\and
{\rm Pietro Michiardi}\\
Eurecom
} 

\maketitle

\thispagestyle{empty}

\subsection*{Abstract}
We consider elastic resource provisioning in the cloud, focusing on in-memory key-value stores used as caches. Our goal is to dynamically scale resources to the traffic pattern minimizing the overall cost, which includes not only the storage cost, but also the cost due to misses.
In fact, a small variation on the cache miss ratio may have a significant impact on user perceived performance in modern web services, which in turn has an impact on the overall revenues for the content provider that uses those services.

We propose and study a dynamic algorithm for TTL caches, which is able to obtain close-to-minimal costs. Since high-throughput caches require low complexity operations, we discuss a practical implementation of such a scheme requiring constant overhead per request independently from the cache size. We evaluate our solution with real-world traces collected from Akamai, and show that we are able to obtain a 17\% decrease in the overall cost compared to a baseline static configuration.

\section{Introduction}

In-memory key-value stores used as caches are a fundamental building block for a variety of services, including web services and Content Delivery Networks (CDN). With the advent of cloud computing, these services, including the caches, have been offered as managed platforms with a pay-as-you-go model. Amazon's ElastiCache \cite{amazonElastiCache} and Microsoft's Azure Redis Cache \cite{azureRedisCache} are examples of caches that employ popular open source software such as Memcached \cite{memcached} or Redis \cite{redis}.

Elasticity, \ie, the ability to adapt to workload changes, is a key characteristic of cloud computing: auto-scaling tools, configured by the users, determine the amount of cloud resources to deploy. The techniques used to drive the scaling process have been the subject of many studies in the past -- see \cite{lorido2014review} and the references therein. These studies mainly focus on traditional services, such as computing, where the relation between the performance and the amount of deployed resources is usually simple. For instance, a highly loaded web server can be duplicated, so that each instance receives half of the traffic, with an almost linear impact on the load.

When considering the caches, the relation between a key performance index, the hit ratio (or, equivalently, the miss ratio), and the resources deployed is not linear, e.g., doubling the cache size does not correspond to doubling the hit ratio (or halving the miss ratio). The analysis of dynamic adaptation of caches has received little attention: the few studies have focused on minimizing storage costs for a given target hit ratio, disregarding the possibility to leverage the hit ratio elasticity itself and ignoring misses costs.

Several studies have highlighted the cost of delay for web services \cite{delayCost}, \ie, a direct connection between the response time (or web page load time) and economic losses, for example because the customer does not finalize a purchase. Notice that, even a small increase in the miss ratio (e.g., 1\%), often translates into a high variation in the average latency (e.g., 25\%) \cite{cidon2015dynacache}. Misses can also translate to infrastructure costs because of the additional load on back-end databases or content servers.  Beyond these specific examples, in this paper we assume that it is possible to quantify the \emph{cost} due to misses. Then, when analyzing  dynamic cache resource allocation, these costs should be considered.

In this paper we study the dynamic assignment of resources for in-memory data stores used as caches. To this aim, we take into account the cost of the storage \emph{and} the cost of the misses, and we adapt the amount of resources to the traffic pattern minimizing the total cost. We consider an approach based on TTL caches \cite{choungmo14}, and we study a model in which the Time-To-Live (TTL) is adapted through stochastic approximation iterations and dynamically converges to the best setting. We operate the system using a virtual TTL cache, whose virtual size informs the elastic deployment of cache server instances to ingest incoming requests.

High-throughput caches rely on low complexity operations: for instance, key lookup and update in LRU caches have $O(1)$ complexity per request. This bound is considered a hard requirement for CDNs running on commodity hardware~\cite{berger2017adaptsize}. The auto-scaling tool, therefore, should not have higher complexity, otherwise it may represent a performance bottleneck. For this reason, based on the results obtained from the model, we design a practical policy to automatically scale horizontally caches, which has $O(1)$ complexity per request.

We evaluate such TTL-based solution with a testbed, using real-world traces collected from production cache servers for over 30 days in Akamai, one of the largest Content Delivery Networks. We show that our approach can achieve the same savings obtained by previously proposed solutions based on Miss Ratio Curves (MRCs) \cite{saemundsson2014dynamic}, which are less scalable because they have a per-request computational overhead that grows logarithmically with the cache size.

\vspace{4mm}
\noindent
{\bf Contributions:} We make the following contributions.
\begin{itemize}
\item {\it TTL-based approach:} We propose and study a dynamic algorithm for TTL caches, which adapts the TTL value to both misses and storage costs minimizing the total operational expenditure (\S\,\ref{sec:adapt-TTL}).
\item {\it Design and implementation of a horizontally scalable TTL-based solution:} We design and implement a system based on the TTL approach, which dynamically adds and removes cache instances in order to maintain the total cost at minimum. We pay particular attention to system scalability, and provide a $O(1)$ solution, as the operations employed in high-throughput caches (\S\,\ref{sec:implementation}).
\item {\it Evaluation:} We evaluate the TTL-based solution in our testbed with real-world traces from Akamai, and show a decrease in the total cost of 17\% compared to the fixed size approach (\S\,\ref{sec:experim-eval}).
\end{itemize}

\noindent
{\bf Roadmap:} In Sect.\,\ref{sec:background} we provide some background, and we define the problem, while in Sect.\,\ref{sec:rel-works} we discuss the related works. We present the general framework of a TTL-based solution in Sect.\,\ref{sec:adapt-TTL}, and discuss its practical implementation in Sect.\,\ref{sec:implementation}. We evaluate our approach in Sect.\,\ref{sec:experim-eval} and conclude the paper in Sect.\,\ref{sec:concl}

\section{Background and problem definition}
\label{sec:background}

\subsection{In-memory data stores}
In-memory key-value stores represent a fundamental piece of web architectures. They are used to cache popular contents, so that the web application can access quickly to the frequently requested data, while the back-end database contains the original copy of all the contents. For instance, Facebook heavily relies on caches based on in-memory data stores, and organize them hierarchically in order to store and access to a complex set of contents \cite{Atik12}.

The most used in-memory stores are Memcached \cite{memcached} and Redis \cite{redis}. While Redis contains a richer set of APIs, when used as a cache it shares with Memcached some basic commands, such as setting a key-value entry, or retrieving the value given a key. If the cache is full and a new content needs to be inserted, both systems employ slight variations of the Least Recently Used policy (LRU). In particular, Memcached organizes the content into classes of objects with similar sizes, and performs LRU within each class. Redis picks randomly 5 objects and evicts the one least recently accessed; if the available space is not sufficient, it repeats the process.

The amount of RAM assigned to Memcached or Redis instances is set when the instance is created, and it cannot be changed at runtime. In order to achieve vertical scalability -- \ie, changing the amount of memory at runtime -- the only option is to create a new instance with the desired amount of memory and transfer the content from the old instance to the new one. Since this approach takes time and resources, vertical scalability is usually not considered practical.

On the other hand, horizontal scalability is easy to achieve. Instances can be added to (or removed from) a cluster of nodes, with a \emph{load balancer} tool (such as \texttt{mcrouter} \cite{mcrouter}) that manages all the aspects related to the distributed caches: data placement and request routing, data replication (if configured) and instance failure management. Data placement and request routing may use consistent hashing to map keys and nodes to points on the hash space, and key responsibility is assigned to the closest node in the hash space.

In this paper, we consider the basic scenario where the content is not replicated across instances and one load balancer is sufficient for managing the cluster. The results can be easily extended to any replication factor that the user may decide to adopt.

\subsection{Elastic on-demand services}
Cloud computing enables services to be instantiated on demand, according to the volume of traffic. In the case of web architectures, for instance, it is possible to modify the number of web servers to accommodate the increasing traffic. Service providers have recently included, among the different services, in-memory data stores used as caches. Prominent examples are Amazon's ElastiCache \cite{amazonElastiCache} and Microsoft's Azure Redis Cache \cite{azureRedisCache}.

These managed solutions take care of the details of the caches, such as software update and maintenance, and provide simple APIs to create and shut down instances, and manage the corresponding cluster of such instances.

The user can choose among a set of possible configurations for each instance. For example, Amazon's ElastiCache \cite{elastiCachePricing} allows the customer to choose among instances with different  RAM size and number of cores (vCPUs). Different types of instance are also available, like regular, spot and burstable ones. The latter two types refer to instances whose capacity may be changed (reclaimed) by Amazon. In this work, we mainly focus on regular instances.

\subsection{Problem definition}
\label{subsec:problem-def}

In this work, we focus on the caches, without considering the other elements of the specific service which exploits the caches, such as the web server, the back-end databases or the origin server if the cache is part of a CDN. Our aim is to adapt over time the total cache size to the content request pattern in order to minimize the total cost, that is the sum of the storage cost and the cost due to the misses.

The storage cost is immediate to evaluate, because it is determined by the pricing scheme of the cloud provider (we provide later specific examples for Amazon ElastiCache service). The provider offers different possible configurations with different costs. As a design choice, we can consider either a set of homogenous instances, or combine heterogeneous instances. The latter option introduces a set of management issues, such as the selection of instances to switch on and off when scaling, which are not simple to deal with. Therefore, we focus on homogenous instances.

Since the cost model of the service providers usually has a specific granularity (typically, one hour), we consider fixed intervals that we call epochs, and the choice of changing the number of instances is done at the end of each epoch.
 The alternative---an asynchronous scenario where instances can be added or removed at any moment---poses more challenges. For instance, it is not simple to decide when addition and removal should be done. Moreover, turning off an instance before the billing period ends is uselessly hurtful  since the user will pay for the full period anyway, therefore desynchronised choices may lead to waste of resources. Let $I(k)$ be the number of instances selected during the $k$-th epoch and $c^{s}$ be the cost of one instance. The storage costs over the first $k$ epochs is then
 \[C^{s}(1,k)= \sum_{h=1}^k c^{s} I(h)\]

Let us now consider the other contribution to the total cost. The cost of a miss can correspond to the additional delay experienced by the final user or to the additional load on the origin server, e.g.~in terms of number of requests or bytes to serve. In any case, we assume indeed that the service provider has quantified monetarily the miss cost. There are for example several studies on the connection between delay and revenues like~\cite{delayCost}. We denote by $m_o$ the miss cost for object $o$, that we assume to be deterministic and constant over time. Let $r(n)$ be the  object requested by the $n$-th miss. With some abuse of notation, we let $n \in [k_1,k_2]$ denote that the $n$-th miss occurred in the time interval corresponding to the epochs $k_1$, $k_1+1$, \dots $k_2$, with $k_1<k_2$. The total miss cost per time unit during the first $k$ epochs is then
\[C^m(1,k)= \sum_{n \in  [1, k]} m_{r(n)}.\]

Our goal is to select the number of instances $I(1)$, $I(2)$, \dots\, $I(k)$, in order to minimize the total cost. The tradeoff is evident. At any epoch, a larger number of instances would decrease the number of misses---and therefore the corresponding cost---but it would cause a higher cost due to storage. Conversely, a smaller number of instances would increase the cost due to misses, but it would decrease the cost of storage. In what follows we present a policy that, at the end of each epoch, determine the number of instances in the cluster such that the (estimated) total cost for the next epoch is minimal.

\subsection{On the complexity of the solution}
\label{subsec:complexity}
In order to deliver high throughput, caches require small processing overheads. For example, $O(1)$ time complexity \emph{per request} is considered a hard requirement for CDN caches~\cite{berger2017adaptsize}. At high request rates, more complex operations can pose an intolerable load on the CPU causing spurious misses~ \cite{neglia17tompecs}, \ie a requested content may not be served even if present in the cache. This is one of the reasons why eviction policies such as LRU and LFU are widely adopted: in fact their operations (key lookup, removal and insertion) have $O(1)$ time complexity. On the contrary, more sophisticated eviction policies proposed in the literature, such as the Greedy Dual Size \cite{cao1997cost} and LRFU \cite{lee1999existence}, may improve over LRU in terms of hit ratio, but have often $O(\log M)$ time complexity per request (where $M$ is the number of objects in the cache) and therefore they pose a high burden on the CPU.

Not only the eviction policy, but any operation related to the cache -- including the operations of the load balancer to route requests and dynamically adapt the cache size -- needs to have small processing overheads. In order to show the impact of the computational complexity on the system, we set up an experiment focusing on the load balancer, where we compare the basic scenario  -- a fixed number of cache instances, where the load balancer simply routes requests to the correct instance using the hash of the key -- with an improved load balancer that computes dynamically the number of instances to use in the cluster. The improved load balancer may use two policies: (i) our TTL-based solution, which has $O(1)$ time complexity, and (ii) a MRC-based solution, which has $O(\log M)$ time complexity (see Sect.\,\ref{sec:rel-works} for a discussion on the complexity of the MRC computation).

For the experiment, we use a trace with more than 2~billion requests collected by Akamai -- the characteristics of the trace and the testbed used are described in Sect.\,\ref{sec:experim-setup}. The trace contains, for each requests, a timestamp. In the first experiment, we replay the trace, \ie, we generate the requests following the timestamps provided in the trace. Figure\,\ref{fig:MRCvsTTL_perf}, left, shows the CPU load over time for two representative days. We see that the additional task to compute the MRC leads to almost double the CPU usage in comparison to the basic scenario, where the load balancer only distributes the requests among a fixed number of instances. On the contrary, the overhead of our TTL-approach is below $20$\%. While the hardware we used in the testbed was adequate to support also the computationally expensive MRC-approach, should the requests rate increase, a scheme with logarithmic complexity would not be able to cope with the processing leading to spurious misses~\cite{neglia17tompecs}.

\begin{figure}[htbp]
\vspace{-5mm}
\hspace{-5mm}
\begin{tabular}{p{0.62\linewidth} p{0.34\linewidth}}
  \vspace{0pt} \includegraphics[width=\linewidth]{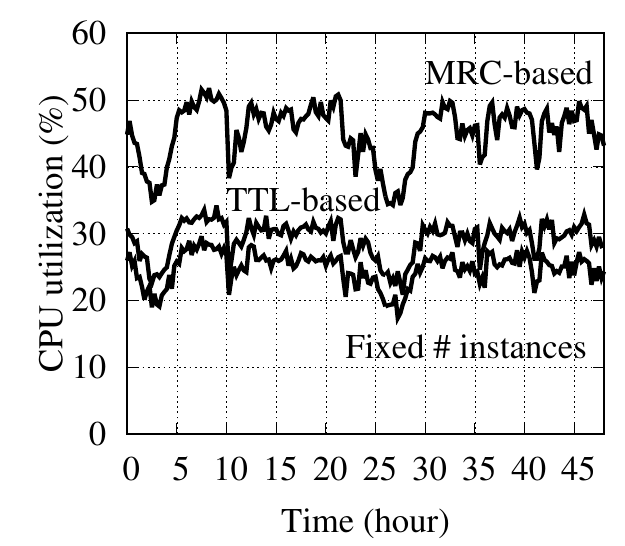} &
  \vspace{0pt} \includegraphics[width=\linewidth]{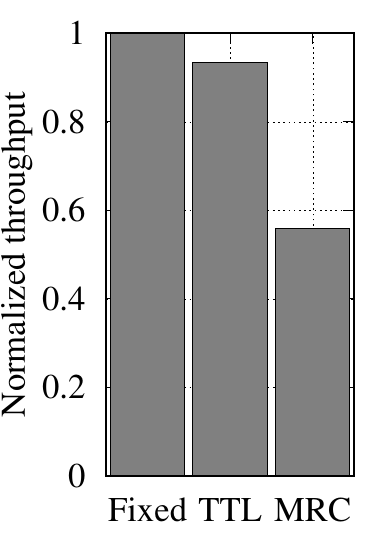}
\end{tabular}
   \caption{Left: CPU load using a fixed requests routing scheme and our TTL-based solution (which have $O(1)$ time complexity) compared to MRC-based solution (which has $O(\log M)$ time complexity). Right: Throughput normalized to the fixed scheme case.}
   \label{fig:MRCvsTTL_perf}
\end{figure}

These findings are confirmed in the second experiment, for which we ignored the trace timestamp, and generated a new request as soon as we received the reply from the load balancer for the previous request. This would provide an indication of the maximum throughput achievable by the different schemes. For ease of representation, we normalize the throughput with respect to the basic scenario with a fixed number of instances. The results are shown in Fig.\,\ref{fig:MRCvsTTL_perf}, right. While our TTL solution experiences about 8\% throughput reduction due to the additional data structure we maintain, the MRC solution almost halves the achievable throughput.

\section{Related works}
\label{sec:rel-works}

Elastic resource provisioning of cloud services has been the subject of many studies. The authors in \cite{lorido2014review} provide a general overview of the techniques, such as control theory as used in \cite{lim2010automated}. Despite the broad set of results and their computational complexity, it is not clear if they can be applied in the context we consider, where the relation between the resource deployed and the key performance index (the hit ratio) is not linear. Moreover, none of them are based on stochastic approximation as our main result on TTL caches.
Another prominent example of a general approach for auto-scaling is given by \cite{shen2011cloudscale} but the proposed solution is based on tools (e.g., time series prediction) that are computationally intensive for the high-throughput scenario we consider.

Similarly to our problem, memory management aims to determine the amount of memory to assign to competing applications, but the proposed solutions, such as \cite{cidon2015dynacache} \cite{carra2014memory} \cite{hu2015lama} \cite{ou2015penalty}, all require computations with higher complexity than our solution.

As for minimizing costs in a cloud computing environment, the authors in  \cite{xu2016blending} and \cite{wang2017exploiting} explore the use of spot instances for different aims, such as content replication or decreasing the overall storage cost. Despite the computational complexity of the solution, the proposed schemes do not take into consideration in the minimization the cost due to misses, as we do. The authors in \cite{wang2017exploiting} also consider a policy for modulating the allocation of on-demand instances to match the dynamic needs, but they do not describe it in detail.

Part of our TTL-based solution is based on the concept of a virtual cache, that maintains the metadata of some cacheable objects, but not their actually content. These objects are sometimes referred to as \emph{ghosts}. This additional information is used in many caching schemes to decide how to manage the objects in the physical cache. For example in 2-LRU~\cite{garetto16} the virtual cache is managed by LRU and a content is actually stored in the physical cache only if its metadata is already present in the virtual cache. As another example, ARC~\cite{megiddo2003arc} uses two virtual caches to decide which contents should be evicted to make space for a new one. Differently from the cases above, we use a virtual cache to size the physical one.

A recent work~\cite{basu2017adaptive} also explores how to adapt the TTL value to the request pattern by using stochastic approximation. In particular, the authors focus on vertical scaling and aim to achieve a target hit ratio, possibly with a small cache size. On the contrary our approach addresses horizontal scaling to minimize the total operational cost.

\vspace{3mm}
\noindent
{\bf MRC-based solutions.} Miss Ratio Curves (MRCs) are a well-known tool for cache profiling \cite{saemundsson2014dynamic}: in a single graph it is possible to observe the relation between cache size and miss ratio, therefore one can compute the cost of the storage and estimate the cost of the misses for each point. The main issue with MRCs is their computational complexity. The seminal algorithm proposed by Mattson \cite{Mattson} takes $O(M)$ operations per request, where $M$ is the number of objects in the cache. A more efficient implementation, proposed by Olken \cite{zhong2009program}, makes use of a tree data structure (e.g., \emph{counting B-Tree}) to keep track of the objects in the cache. This reduces time complexity to $O(\log M)$ per request. In order to decrease the complexity to $O(1)$ per request, many solutions have been proposed in the literature that compute approximate MRC~\cite{saemundsson2014dynamic} \cite{wires2014characterizing} \cite{waldspurger2015efficient} \cite{waldspurger2017cache}. Such solutions share a common characteristic: they have been designed considering objects with uniform sizes. On the contrary, the applications we are interested into exhibit contents with \emph{heterogeneous sizes}.

We observe that it is possible to extend Olken's approach to MRC computation to the case of heterogeneous size contents maintaining $O(\log M)$ complexity per request. To this aim, we suggest to use a special tree, called \emph{order statistics tree}, that has a method \texttt{rank($x$)}, which returns the sum of the weights of the elements with keys less than or equal to $x$ (the weights are the object size).\footnote{This is how we compute MRCs in this paper. We suspect that this approach may be known, but we were not able to find it described elsewhere.} On the contrary, it is not clear how to extend the algorithms to compute approximate MRCs, while maintaining the same level of accuracy. We support this claim with the following experiment.
We consider in particular the method proposed in~\cite{waldspurger2015efficient} \cite{waldspurger2017cache} (but the others operate in a similar way): the request trace is sampled with rate $R$, a first MRC is computed on the subsampled trace and then scaled it up opportunely to get the approximate one for the whole traffic. A constant sampling rate would lead to $O(\log R M)$ time complexity. By adapting dynamically $R$ according to the available memory, it is possible to reach $O(1)$.

We use a trace described in \cite{neglia17tompecs}, \S\,4.4 -- the distribution of the object popularity and object size are also reported in Fig.\,\ref{fig:trace_intput}, but for understanding the results of the experiment, the details of the trace are not essential. For each request, besides the timestamp and the object identifier, we have the object size. First, we ignore the actual object size and assume it to be uniform. In this case, the method predicts the MRC with a prediction error\footnote{
	As in~\cite{waldspurger2015efficient}, the error is evaluated by measuring the absolute difference between the exact and the approximated MRCs over all the meaningful cache sizes, and then by computing the mean of these absolute differences.}
	smaller than $3*10^{-3}$ for all sampling rates between  0.1 to 0.001---see Fig.\,\ref{fig:mae_Het_hist}---, similarly to what observed in the original papers \cite{waldspurger2015efficient} \cite{waldspurger2017cache}. We then consider the object sizes and repeat the experiments: for a given sampling rate, the error increases by one order of magnitude! Moreover, in order to reach a given target error it may be necessary to increase by two orders of magnitude the sampling rate. Correspondingly, the dynamic sampling rate approach described in~\cite{waldspurger2015efficient} would require a larger memory footprint and a larger number of operations for request.

\begin{figure}[htbp]
   \centering
   \includegraphics[width=0.8\linewidth]{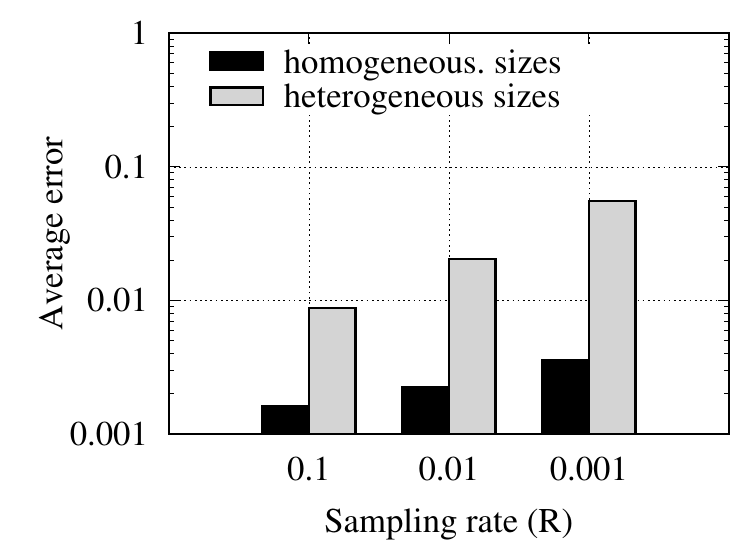}
   \caption{Accuracy of the approximate MRC computation through sampling, with uniform and nonuniform object sizes.}
   \label{fig:mae_Het_hist}
\end{figure}

This simple experiment shows that object sizes may have an unexpected impact on the accuracy of the approximated MRC computation. Other approximated techniques, designed to have $O(1)$ time complexity per request, such as MIMIR \cite{saemundsson2014dynamic} or AET \cite{hu2016kinetic}, may be affected when considering a scenario with heterogeneous object size. In MIMIR, the authors state that their solution can be extended and consider non uniform sizes, but they do not show experimentally their claim, and their solution is not publicly available. Therefore, it is not possible to understand the impact of heterogeneous object size on the accuracy of their method. Similarly, the AET approach assumes uniform object size, and it is not clear if it can be easily extended to consider heterogeneous object size.

In summary, the approximated computation of the MRC, even if it can be performed with $O(1)$ time complexity per request, still need to be studied in depth, especially for the heterogeneous size case. Currently, therefore, the only option is to compute the MRCs exactly, which has $O(\log M)$ complexity per request.

\section{Adaptive TTL based solutions}
\label{sec:adapt-TTL}

In this section, we begin with a key building block to the design of a horizontally scalable caching system. Our work draws inspiration from Time-To-Live (TTL) caches, i.e.~caches that are managed by a TTL policy. There are two families of TTL policy: with and without renewal. In both cases, upon a miss, the content is stored locally and a timer with duration equal to $T$ is activated and the content is evicted when the timer expires. The difference is that, in the case with renewal, the timer is reset by the following hits for the content, while it is not affected by them in the case without renewal. TTL caches are a natural model for DNS caches, but they have also been proposed as an approximated model to study the performance of existing replacement policies like LRU~\cite{che02}. Moreover, different papers have suggested their practical use because of their higher configurability as well as amenability to analysis~\cite{choungmo14,dehghan16,basu2017adaptive}. While a replacement policy maintains in the cache as many contents as the available space buffer allows (contents are evicted only if needed to make space, under a TTL policy the actual storage keeps varying over time and is, in theory,  potentially unbounded. A real implementation of a TTL cache will have a finite capacity and then it may need to evict some contents from the cache even if their timer has not expired yet. Some of these practical issues are discussed in~\cite{choungmo14}. In our solution a TTL cache with renewal is used as a virtual cache, storing only content metadata:\footnote{The total storage required by the virtual cache will then be negligible.} by computing its virtual size, our approach steers the addition or removal of cache server instances.

\subsection{Dynamic adaptation}
We present an adaptive mechanism based on stochastic approximation by which the timer value converges to the value that minimizes the total cost.

The theoretical results hold in the following scenario. We consider a finite catalogue with $N$ contents and that requests for the different contents occur according to independent renewal processes. We denote by $\lambda_i$ the request rate for content $i$. A case of particular interest in what follows is the case where these processes are Poisson ones. Then, a given request will be for content $i$ with probability $\lambda_i / \sum_{j=1}^N \lambda_j$ independently from any previous request. This is (a continuous version of) the well known Independent Reference Model (IRM) \cite{coffman1973operating}.

In what follows, we consider an ideal TTL cache with renewal and assume that the cloud service charges the user only for the instantaneous storage occupancy. This differs from the more realistic scenario described above where the user needs to pay for the instances independently from their usage, but we will come back to the more realistic billing in Sect.\,\ref{sec:ttl_practical}.
Let $s_i$ be the size of object $i$ and $c$ be the cost per unit time to store a unit of content (\cite{wang2017exploiting} shows that prices are almost linear also for real cloud services).
Then, the total cost to store content $i$ over a time window of duration $\tau$ is $c s_i \tau$. For simplicity, we denote $c_i = s_i c$. A miss for content $i$ incurs a cost equal to $m_i$.

Let $X_i(t)$ be the indicator function for the event ``content $i$ is stored in the cache at time $t$'' and $M_i(t)$ the counting process of content $i$ misses in the interval $[0,t]$.
We can define the storage cost and the miss cost analogously to what done in Sect.\,\ref{subsec:problem-def}.  The total cost  over the interval $[0,t]$ is then
\begin{align}
\label{e:cost_ttl}
 C(0,t) & = C^s(0,t) + C^m(0,t)\nonumber\\
& = \sum_{i=1}^N \int_0^t X_i(t) c_i  \diff t + M_i(t) m_i.
\end{align}

If the caching policy uses a constant TTL value equal to $T$, then each process $X_i(t)$ is a renewal process whose regeneration points are the time instants at which content $i$ misses occur. The renewal reward theorem guarantees that, for each content, the time-average cost is equal to the expected cost over a renewal period divided by the expected duration of a renewal period, \ie
\[\lim_{t \to \infty}\frac{\int_0^t X_i(t) c_i  \diff t + M_i(t) m_i}{t} = \frac{c_i \tau_{S,i} + m_i}{\tau_{M,i}},\]
where $\tau_{S,i}$ is the expected sojourn time of content $i$ in the cache and $\tau_{M,i}$ is the expected time between two misses.

The asymptotic time average cost ($\mathcal C$) of the system as a function of $T$ is then
\begin{equation}
\label{e:cost_renewal}
\mathcal C(T) = \lim_{t \to \infty} \frac{C(0,t)}{t}= \sum_{i=1}^N \frac{c_i \tau_{S,i} + m_i}{\tau_{M,i}}.
\end{equation}
We observe that $\tau_{S,i}/\tau_{M,i}$ is the asymptotic fraction of time content $i$ spends in the cache or equivalently, the probability that content $i$ is in the cache at a random time, that is often called the occupancy probability and we will denote by $o_i$. The inverse of $\tau_{M,i}$ is the rate at which miss occurs that we can also write as $\lambda_i (1- h_i)$, where $h_i$ is the hit ratio, \ie the fraction of requests for content $i$ that incurs a hit. Then we can rewrite \eqref{e:cost_renewal} as
\begin{equation}
\label{e:cost_renewal2}
\mathcal C(T) = \sum_{i=1}^N c_i o_i + \lambda_i m_i (1-h_i).
\end{equation}

When the arrival process is IRM, it holds $o_i=h_i$ because of PASTA property and moreover $h_i=1-e^{-\lambda_i T}$~\cite{choungmo14}. \eqref{e:cost_renewal2} becomes
\begin{equation}
\label{e:cost_irm}
\mathcal  C(T) = \sum_{i=1}^N  c_i + (\lambda_i m_i -c_i) e^{-\lambda_i T}.
\end{equation}
We can check that if $T=0$, \ie no content is stored in the cache,  the cost per time unit is equal to $\sum_{i=1}^N \lambda_i m_i$: we pay systematically for all the misses. Instead, if $T_i=\infty$, all the contents are stored indefinitely and the corresponding cost per time unit is $\sum_{i=1}^N  c_i$.

We could look for the value $T^*$ that minimizes the cost~\eqref{e:cost_irm} by applying a gradient algorithm as follows:
\begin{align*}
T(n+1) & = T(n) - \epsilon(n) \frac{d \mathcal C}{d T}\Big|_{T(n)} \\
	& =T(n) + \epsilon(n) \sum_{i=1}^N \lambda_i e^{-\lambda_i T(n)} \left( \lambda_i m_i - c_i \right),
\end{align*}
where the sequence $\epsilon(n)$ converges to zero as $n$ diverges, but it is not summable, \ie $\sum_{n \in \mathbb N} \epsilon(n) = \infty$.
This approach is not viable because in a realistic scenario popularities are unknown, keep changing over time and are not easy to estimate. The gradient algorithm suggests us a practical solution based on stochastic approximation~\cite{kushner03}. We observe that $\lambda_i e^{-\lambda_i T}= \lambda_i (1-h_i)$ is equal to the miss rate for content $i$. Upon a miss, this is for content $i$ with probability proportional to $\lambda_i (1-h_i)$.
Let $r(n)$ be the object requested at the $n$-th miss and $\hat \lambda_i(n)$ be an unbiased estimate of the arrival rate $\lambda_i$. Consider the following update rule for the variable $T(n)$:
\begin{equation}
\label{e:ttl_update}
T(n+1) = T(n) + \epsilon(n) \left( \hat \lambda_{r(n)} m_{r(n)} - c_{r(n)} \right),
\end{equation}
where the correction term $ \hat \lambda_{r(n)} m_{r(n)} - c_{r(n)} $ is a random variable because i)  content requests occur according to IRM and ii) the estimator itself is a random variable.
The correction corresponds ``on average'' to the gradient $d\mathcal  C/dT$ because, upon a miss, the fraction of requests for content $i$ is proportional to $\lambda_i e^{-\lambda_i T(n)}$, and then $\E(\hat \lambda_i m_i - c_i)= \lambda_i m_i - c_i$. The following proposition makes this result formal.

\begin{proposition}
\label{p:stoch_approx}
Let $\{X(n,T(n))\}$ be a sequence of independent random variables such that $X(n,T(n))$ is equal to $\hat \lambda_i m_i - c_i$ with probability $\lambda_i e^{-\lambda_i T(n)} / (\sum_{j=1}^N \lambda_j e^{-\lambda_j T(n)})$. Let $\{\epsilon(n)\}$ be a non-negative sequence converging to $0$, such that $\sum_{n\in \mathbb N} \epsilon(n) = \infty$ and $\sum_{n\in \mathbb N} \epsilon^2(n) < \infty$. Consider the update rule
\[T(n+1) = \Pi_{[0,T_{\max}]}\! \left(T(n) + \epsilon(n) X(n,T(n))\right),\]
where $\Pi_{[0,T_{\max}]}\!(x)=\min(\max(0,x),T_{\max})$ is the projection operator over the interval $[0,T_{\max}]$, then the sequence $T(n)$ converges with probability one to i) a stationary point of $\mathcal C(T)$ or ii) $0$ or $T_{\max}$, if $0$ and $T_{\max}$ are local minima of $\mathcal C(T)$.
\end{proposition}
\begin{proof}
The result follows from Theorem 2.1 in~\cite{kushner03}. All the hypotheses $(\mbox{A}2.1)-(\mbox{A}2.7)$ are satisfied with $f(.)=\mathcal C(.)$.
\end{proof}

If, instead of letting the weights $\epsilon(n)$ converge to zero, we keep them equal to a small constant value $\epsilon_0$, then, in a stationary setting, $T(n)$ converge to a neighbourhood of the limits indicated in Proposition~\ref{p:stoch_approx}. At the same time, a constant weight makes it possible to track changes in the system, for example when popularities keep varying over time.

\subsection{An optimal clairvoyant TTL Policy}
\label{sec:ttl-opt}
In this section we present the optimal TTL policy (referred to as TTL-OPT), that minimizes the total cost when the sequence of future requests is known. The cost achieved by this clairvoyant policy is clearly a lower-bound for any feasible policy.
Among the TTL policies, TTL-OPT plays the same role as B\'el\'ady's algorithm~\cite{belady1966study} for replacement policies. Indeed, B\'el\'ady's algorithm minimizes the miss ratio under knowledge of the future requests and uniform content sizes. Interestingly, the optimal clairvoyant TTL policy has polynomial complexity under heterogeneous content sizes and miss costs, while in such case the Belady's policy is no more optimal: finding an optimal replacement policy is an NP-complete problem~\cite{hosseini2000optimal}, due to a hard capacity constraint, that makes the problem intrinsically combinatorial.

\begin{algorithm}
\SetKwInOut{Input}{input}\SetKwInOut{Output}{output}
\DontPrintSemicolon
\SetAlgoVlined
\Input{$\{c_i\}$, storage costs per unit of time}
\Input{$\{m_i\}$, miss costs}
\Input{request sequence}

\BlankLine
	\ForEach{request $r$}{
		$j \leftarrow \mbox{obj id of request }r$\;
		$t_{j,\mbox{\scriptsize next}} \leftarrow \mbox{time of the next request for obj }j$\;
		$c^S_j \leftarrow c_j \times  \left( t_{j,\mbox{\scriptsize next}} -  t_{\mbox{\scriptsize now}} \right)$\;
		\eIf{($c^S_j < m_j$)}{
			$T_j \leftarrow t_{j,\mbox{\scriptsize next}} -  t_{\mbox{\scriptsize now}}$ \tcp{store $j$ until its next request}
		}{
			$T_j \leftarrow 0$ \tcp{do not store $j$}
		}
	}
	\caption{Optimal Clairvoyant TTL policy (TTL-OPT)}
	\label{alg:cost_opt}
\end{algorithm}

We allow the optimal policy TTL-OPT to select a TTL value different for each content and for each request. The policy is described in Algorithm~\ref{alg:cost_opt} and is very simple: given a request for a content, say $j$, at time $t_{\mbox{\scriptsize now}}$, if the cost to store the content until its next request (at time $t_{j,\mbox{\scriptsize next}}$) is smaller than the cost of a miss for this object, then the content should be stored in the cache until the next request, i.e. the timer should be set equal to $t_{j,\mbox{\scriptsize next}}- t_{\mbox{\scriptsize now}}$. Otherwise, the object should be served but not stored.
The formal proof of TTL-OPT follows.
\begin{proposition}
The clairvoyant policy TTL-OPT in Algorithm~\ref{alg:cost_opt} minimizes the sum of storage and miss costs.
\end{proposition}
\begin{proof}
Let $C_i(0,t)$ denote the total cost paid during the interval $[0,t]$ for content $i$, i.e.
\[C_i(0,t)=\int_0^t X_i(t) c_i  \diff t + M_i(t) m_i.\]
The total cost $C(0,t)$ in~\eqref{e:cost_ttl} is then given by the sum of the costs for each content.
The possibility to choose the timer value independently for each content reduces the minimization of the total cost $C(0,t)$ to separately minimize each term $C_i(0,t)$.

Let $\{t_{i,k}, k \in \mathbb N\}$ be the sequence of time instants of the requests for content $i$. A TTL policy needs to select a TTL value for each request, let us denote as $T_{i,k}$ the timer for the $k$-th request occurring at time $t_{i,k}$.
We observe that we can restrict ourselves to consider $T_{i,k}\in \{0, t_{i,k+1}- t_{i,k}\}$. In fact, consider any sequence of timer values $\{\hat T_{i,k}, k \in \mathbb N\}$, and let $\hat T_{i,h}$ be a timer such that $\hat T_{i,h} < t_{i,h+1}- t_{i,h}$. If we replace $\hat T_{i,h}$ with $T_{i,h}=0$, the cost $C_i(0,t)$ cannot increase. Similarly, we can replace any value $\hat T_{i,h}$ such that $\hat T_{i,h} > t_{i,h+1}- t_{i,h}$ with $T_{i,h}= t_{i,h+1}- t_{i,h}$, without increasing the cost $C_i(0,t)$.
Let then $Z_{i,k}$ be an indicator function such that $Z_{i,k}=1$ if $T_{i,k}=t_{i,k+1}- t_{i,k}$, and  $Z_{i,k}=0$ if $T_{i,k}=0$. The total cost for content $i$ can then be rewritten as follows:
\begin{equation}
\label{e:cost_linear_prog}
C_i(0,t_{i,k})=m_i+\sum_{h=0}^{k-1} \left( Z_{i,h} c_i (t_{i,h+1}- t_{i,h}) + (1- Z_{i,h}) m_i  \right),
\end{equation}
where the first term on the right hand side corresponds to the fact that the first request for content $i$ generates always a miss. From~\eqref{e:cost_linear_prog} if follows that $C_i(0,t_{i,k})$ is minimized by choosing $Z_{i,h}=1$ if $c_i (t_{i,h+1}- t_{i,h}) < m_i$ and $Z_{i,h}=1$ otherwise. This corresponds to what TTL-OPT does.
\end{proof}

Clearly, the TTL-OPT policy can not be used online. Nevertheless, given a trace, its cost can be computed (in polynomial time) and used as a reference.

\section{Implementation}
\label{sec:implementation}

In this section we present a practical, efficient implementation of a TTL cache with $O(1)$ complexity. Then, we describe the operation of our elastic caching system, by focusing on the load balancer algorithm that determines the total cache size, and hence the storage cost.

\subsection{Practical implementation of the TTL-based scheme}
\label{sec:ttl_practical}
A straight application of~\eqref{e:ttl_update} would require to update the timer immediately upon a miss, and then popularity estimates should be available for contents that are not in the cache. Instead, we will start estimating content popularity immediately after the content is stored in the cache and we will then postpone the timer update to the moment when the estimate is available. The detailed description follows.

Let $T(t)$ be the timer value at time $t$. If the timer is updated at $t$, then we denote as $T^-(t)$ the value immediately before  the update. Updates are, as above, driven by misses, and we denote as $t_n$ the time of the $n$-th miss and $r(n)$ the corresponding content. Upon a miss, content $r(n)$  is stored and its timer is set to the current value $T(t_n)$. Any new request for content $r(n)$ before the timer expiration will be a hit and will reset the timer to $T(t_n)$. The number of hits for content $r(n)$ during the interval $[t_n, t_n+T(t_n)]$ is recorded. Let us denote this number as $h_{r(n)}$. The ratio $h_{r(n)}/T(t_n)$ is an unbiased estimator of the rate $\lambda_{r(n)}$. Once this estimate is available at time $t_n+T(t_n)$, the timer is updated as follows:
\begin{align}
T(t_n+T(t_n)) = & T^-(t_n+T(t_n)) \nonumber\\
& + \epsilon(n)\left( -c_{r(n)} + \frac{h_{r(n)}}{T_{r(n)}}m_{r(n)}\right).
\label{eqn:ttl_update}
\end{align}
We observe that in general $T^-(t_n+T(t_n))$ is different from $T(t_n)$, because the timer may have been updated during $[t_n, t_n+T(t_n)]$ as effect of misses for contents other than $r(n)$.

As a further refinement, we notice that the cache is driven by two main events: request arrival and object eviction. The updates of the timer should be done during these events, so that we do not need to create a specific event for each miss for updating the timer. This adds an additional small delay, as it is shown in Fig.~\ref{fig:ttl_update}: given a content, the TTL update is triggered by the hit after the first timer (case \emph{a}), or, if no hit occurs after this time, by the content eviction (case \emph{b}).

\begin{figure}[htbp]
   \centering
   \includegraphics[width=0.9\linewidth]{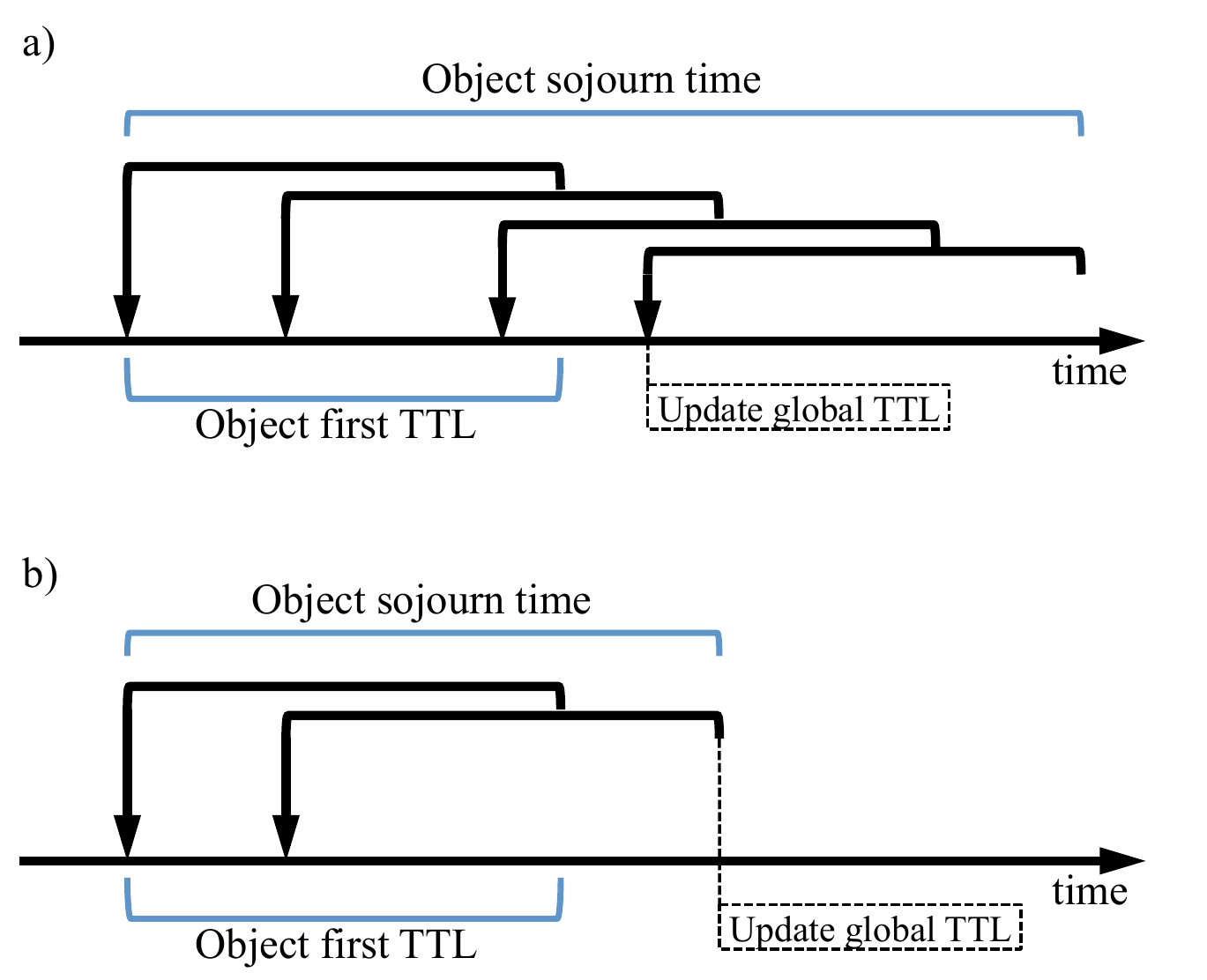}
   \caption{Global TTL update.}
   \label{fig:ttl_update}
\end{figure}

The update rule in \eqref{eqn:ttl_update} -- with the additional considerations shown in Fig.~\ref{fig:ttl_update} -- leads to a feasible implementation.
We observe that Proposition~\ref{p:stoch_approx} does not hold for this new algorithm for two reasons. First, it is not true that the different updates are independent and identically distributed (conditioned on the current timer value). For example, upon a miss for content $i$, it is less likely that the following miss will also be for content $i$, because the content was stored in the cache right after the first miss. Second, the update delays could in principle affect convergence. There are theoretical results for stochastic approximation algorithms when the correction terms are correlated and when updates are delayed, and we indeed think that the implementation described above may still converge, but we leave this study for further investigation, also because \eqref{eqn:ttl_update}, while implementable, would not satisfy our requirement on $O(1)$ computational complexity as we are going to discuss.

We observe that~\eqref{eqn:ttl_update} would require to store in a calendar the expiration instants of the timers, both to update the TTL values and to evict the contents whose timer expired without being renewed by another request. Unfortunately, the expirations instants are not necessarily ordered as the request arrival times: while it always holds $t_n > t_{n-1}$, it may be $t_n + T(t_n)<t_{n-1} + T(t_{n-1})$, because the TTL value may have been decreased between $t_{n-1}$ and $t_n$. As a consequence, the calendar would require a data structure that enables ordered data insertion, which has $O(\log M)$ complexity, where $M$ is the number of contents stored in the cache.

In order to maintain a $O(1)$ cost, we implement the calendar as a FIFO, so that events are orderly inserted at the head according to the request arrival times. For the eviction, we process the events from the tail until a non-expired event is found. The drawback of this approach is that objects that should be evicted because their timer is expired may persist for some additional time in the cache. Nevertheless, the impact of these objects on the overall performance should be negligible. To verify this claim, we have conducted extensive experiments, where we compare the TTL based solution corresponding with~\eqref{eqn:ttl_update} with our solution achieving $O(1)$ complexity, and we observed no significant difference in terms of TTL, instantaneous cache size, or final cost .

\subsection{Horizontally scalable cache system}
\label{sec:horizontal_system}

The TTL-based scheme discussed so far considers a single TTL cache, where the instantaneous  storage occupancy is billed. In other words, we have considered a perfect vertically-scalable system, where memory resources can be smoothly added and removed. In this section we discuss the design of a more practical horizontally-scalable system inspired by the TTL-based approach, where storage can only change at finite epochs.

In a horizontally-scalable solution, cache instances can be added or removed from the cluster, and all the instances have the same configuration. The first design choice to face is the configuration of a generic instance.

\vspace{1mm}
\noindent
{\bf Cache instances:}
These are the physical caches storing the actual contents and have fixed size. They can be implemented using Memcached or Redis with a simple eviction policy like LRU.

\vspace{1mm}
\noindent
{\bf Load balancer:}
The load balancer performs the ordinary operations, such as request routing, and content insertion, \ie, in case of a miss, after retrieving the object from the origin or the back-end, it stores it in one of the cache instances. In addition, the load balancer maintains a virtual cache, with the references of the requested objects: this virtual cache is going to be managed as a TTL cache according to the description in Sect.\,\ref{sec:ttl_practical} and then with $O(1)$ computational cost per  request. The size of the virtual cache depends on the timer value $T$, which in turn depends on the number of hits and misses, and on the corresponding costs for the storage and for the misses. Thus, the size of the virtual cache summarizes the influence of costs of the previous requests, and can be used to determine the number of actual instances to employ in the cluster.

\vspace{1mm}
\noindent
{\bf Proposed scheme:}
Our scheme is described in Algorithm\,\ref{alg:ttl_based}. At every request, we look for the object key in the virtual cache, update its position in case of a hit, or add it in case of a miss. Then, we start evicting objects from the virtual cache if they are expired. While inserting a new object or removing expired objects, we update the total size of the cache (the sum of the sizes of non-expired objects). Clearly, object sizes can be heterogeneous. At the end of the epoch (line \ref{algline:exp}), we look at the size of the virtual cache and we select the number of instances such that the sum of the sizes is the closest to the virtual cache size (line \ref{algline:update}).

At the end of the observation interval, if the number of instances has changed, the load balancer reassigns the responsibility of the hash space to the current instances.

\begin{algorithm}
\SetKwInOut{Input}{input}\SetKwInOut{Output}{output}
\DontPrintSemicolon
\SetAlgoVlined
\Input{VC, Virtual Cache}
\Input{$S_p$, Physical cache size}
\Output{$I(k+1)$, \# of the instances in $k+1$-th epoch}
\BlankLine
	\ForEach{request $r$}{
		\If{($r \in$ VC)}{
			\textsc{Remove}($r$, VC)\;
		}
		$r.expire \leftarrow t_{\mbox{now}} + TTL_{\mbox{now}}$\;
		\textsc{Insert}($r$, VC)\;
		\textsc{EvictExpired}(VC)\;
	}
	\If{(epoch $k$ ended) \label{algline:exp}}{
	 $I(k+1) \leftarrow$ \textsc{Round}(VC.size / $S_p$)\; \label{algline:update}
	}
	\caption{TTL-based scaling}
	\label{alg:ttl_based}
\end{algorithm}

\vspace{1mm}
\noindent
{\bf Additional considerations:}
We observe that objects stored in the physical caches may be different from the ones maintained by the virtual cache. When a physical cache needs to make space for new data, it may evict one content before the timer of the corresponding ghost at the virtual cache is expired. On the other side, the eviction of the ghost does not cause the eviction of the actual content at the cache. Moreover, when instances are added or removed from the cluster, the object key space responsibility must be rearranged, which may lead to a spurious misses due to route changes: the object is in a physical cache, but the request is routed to a different one. Overall, therefore, we will observe virtual misses at the virtual cache, and actual misses at the physical cache, and these two values may be different. We have experimentally observed that, since the number of requests within an epoch is usually very high, the effect of spurious misses due to the change of the number of instances is negligible.

\section{Experimental evaluation}
\label{sec:experim-eval}

\subsection{Setup}
\label{sec:experim-setup}
We evaluate our approach using a \emph{testbed} that is representative of a typical web architecture.

\vspace{1mm}
\noindent
{\bf Testbed:} An application server is connected to a database and to a cluster of caches. The application receives the requests and checks if the content is stored in the cache. If the content is not in the cache, the application server retrieves the object from the database, serves the client and stores the object in the cache. For the operations related to the cache (e.g., object lookup), the application server relies on a \emph{load balancer}. We have implemented the scheme described in Sect.\,\ref{sec:horizontal_system} in a custom tool similar to \texttt{mcrouter} \cite{mcrouter}. The tool is able to add or remove the cache instances from a local cluster, but it can be easily extended so that it can use the APIs of a cloud cache provider.

\vspace{1mm}
\noindent
{\bf Trace:} The requests sent to the application server are generated by reproducing two anonymized traces collected for over 30 days and 5 days from two vantage points of the Akamai network. The traces we tested contain the timestamp of the request arrival, the anonymized object ID and the size of the requested object. We report the results for the 30-day trace, since we obtain similar qualitative results with the 5-day trace. In the 30-day trace, there are $2\cdot10^9$ requests for 110 millions contents, whose size varies from few bytes to tens of MB. Figure\,\ref{fig:trace_intput} (left-hand side) shows the number of requests for each object, sorted by rank (in terms of popularity). The right-hand side shows the empirical Cumulative Distribution Function (CDF) for the size of the requested objects (without aggregating requests for the same object). We could not carry on our experiments with other traces, like the ones collected in the public repository \cite{sniaRep}, because they refer to low level storage (block I/O), and they are not representative of a typical cloud-based application, such as a web service.

\begin{figure}[htbp]
   \centering
   \includegraphics[width=0.48\linewidth]{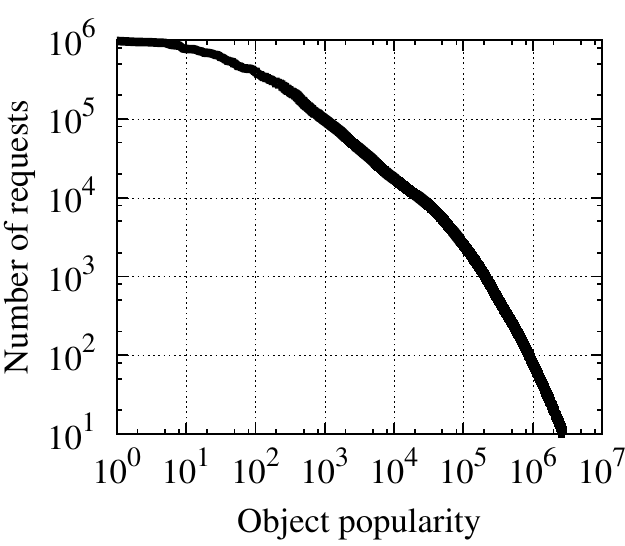}
   \includegraphics[width=0.48\linewidth]{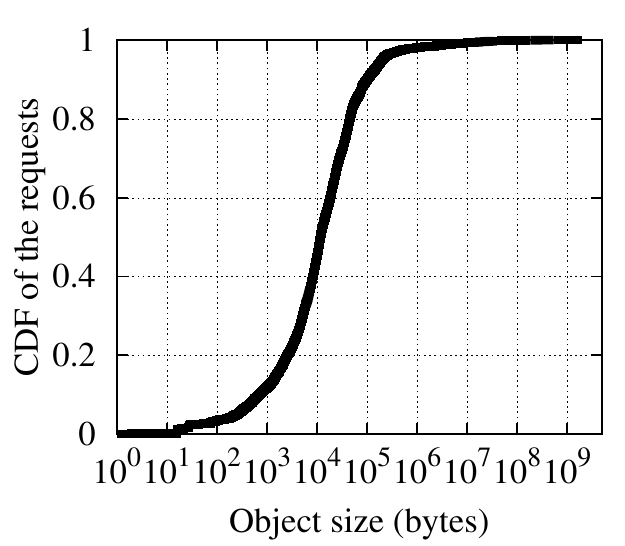}
   \caption{Number of requests per object, ordered by rank (left), and cumulative fraction of the requests for objects up to a given size (right).}
   \label{fig:trace_intput}
\end{figure}

\vspace{1mm}
\noindent
{\bf Settings:} For the configurations and the costs, we refer to Amazon ElastiCache service \cite{elastiCachePricing}. For the duration of the epoch, we consider the minimum billing time, which is one hour. Among the different instances' options, we selected the \emph{cache.t2.micro} with 0.555 GB RAM and one vCPU, which costs 0.017\$/hour (Oct. 2017, US based). We use a small instance since it provides a fine granularity when we resize the cluster: for instance the experimental results in Fig.\,\ref{fig:ttl_cacheSize_trace} shows that one small instance is sufficient during low traffic periods. Moreover, bigger instances (e.g., with 3.22 GB or 6.05 GB) have just two vCPUs, which may limit the throughput of the cache. Replicating small instances each with a vCPU helps in maintaining the throughput while scaling the cluster. As for the cache, we use Redis rather than Memcached in order to avoid problems related to calcification \cite{carra2014memory} \cite{hu2015lama} \cite{ou2015penalty}.

In order to determine reasonable miss costs, we reasoned as follows. The production server from which our trace was collected had an in-memory cache of $4$ GB~\cite{berger2017adaptsize}, i.e.~roughly corresponding to eight \emph{cache.t2.micro} instances. We assume that this is a well engineered system whose cache size has been selected so that storage and miss costs are equal. This is a reasonable rule of thumb to select a close-to-optimal size.
The storage cost can be determined in our case considering the corresponding hourly cost of eight \emph{cache.t2.micro} instances. By dividing this cost by the average number of misses observed during one our in production, we obtain the cost per miss
(in our case, 1.4676 $\times 10^{-7}$ \$ per miss).

\vspace{1mm}
\noindent
{\bf Previous solutions:} Because of the considerations above, we consider as baseline setting a scenario with eight \emph{cache.t2.micro} instances. We compare also our results with an elastic resource allocation scenario driven by the MRC-approach, as described in \cite{saemundsson2014dynamic} and discussed in Sect.\,\ref{sec:rel-works}. In addition, as a reference, we consider the scenario with an ideal, vertically scalable, pure TTL cache, billed according to its instantaneous size.

\subsection{Results}
We present here the results for the trace described in Sect.\,\ref{sec:experim-setup}. We have also performed an extensive study using synthetic traces generated according to the IRM model -- which is the arrival pattern for which the theoretical results in  Proposition~\ref{p:stoch_approx} hold. In such experiments, it is possible to see that the TTL indeed reaches a stable value, which corresponds to the minimum cost.

With a real trace, the arrival pattern varies over time. Our TTL approach continuously tracks such a variation: this is shown in Fig.\,\ref{fig:ttl_cacheSize_trace} (left), where we plot the value of the TTL for an interval of four representative days: the evolution clearly follows a daily pattern. The fluctuation of the TTL is mirrored by the virtual cache size (Fig.\,\ref{fig:ttl_cacheSize_trace}, right), which varies from zero (the cost of the few misses does not justify the storage of the object) to 3.5 GB.

\begin{figure}[htbp]
   \centering
   \hspace{-2.0em}
   \includegraphics[width=0.53\linewidth]{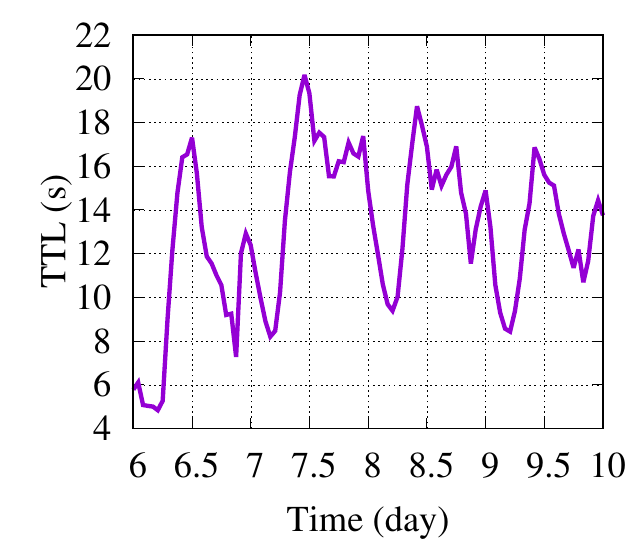}
   \hspace{-1.0em}
   \includegraphics[width=0.53\linewidth]{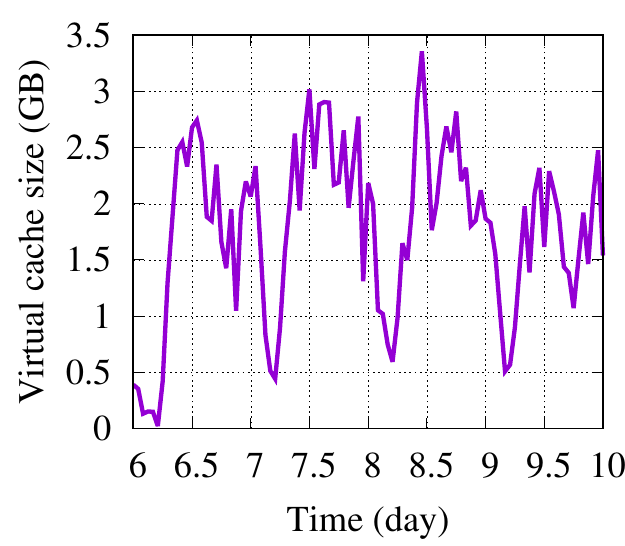}
   \caption{Virtual cache: TTL over time (left), and cache size (right).}
   \label{fig:ttl_cacheSize_trace}
\end{figure}

The virtual cache size translates into the number of instances used in the cluster. From this, it is possible to compute the total cost for storage and misses. In Fig.\,\ref{fig:dyn_cumulcost} we show the cumulative costs for the first 15 days (along with a zoom at day 15) for the TTL-based system, and we compare it with a 8-instance fixed-size cache (corresponding to our reference in-memory production cache) and the MRC-based approach (discussed in Sect.\,\ref{sec:rel-works}). The figure plots also the total cumulative cost of an ideal TTL-cache. The results show that the TTL-based approach obtains similar cumulative costs as the MRC-based approach, but with a $O(1)$ complexity instead of $O(\log M)$ complexity. Overall, with respect to the baseline fixed-size approach, the TTL-based approach is able to save 17\% of the costs.
The difference between the ideal and the practical TTL-based implementation is due to the discretization of cache sizes and billing periods, and the spurious misses due to the reorganization of the object key responsibility. Nevertheless, such a difference causes only a  2\% cost increase in comparison to the ideal implementation. Interestingly, this result suggests that, at least for typical CDN applications, there is no need for finer-grained billing periods or cache sizes, but most of the potential improvement is already achievable with the current offer.

\begin{figure}[t]
   \centering
   \hspace{-2.0em}
   \includegraphics[width=0.53\linewidth]{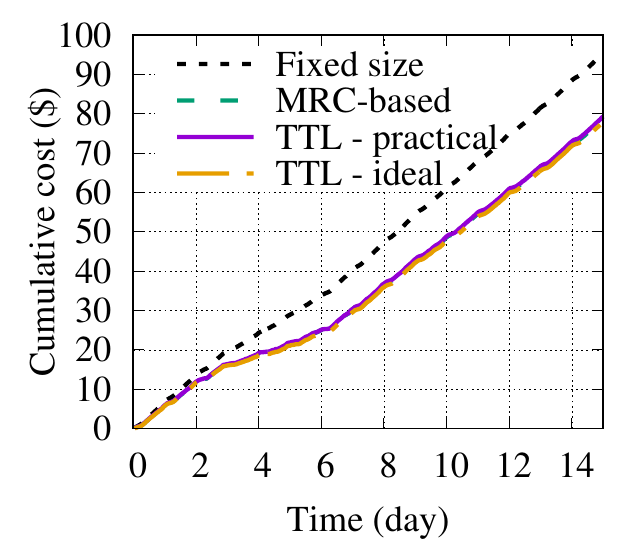}
   \hspace{-1.0em}
   \includegraphics[width=0.53\linewidth]{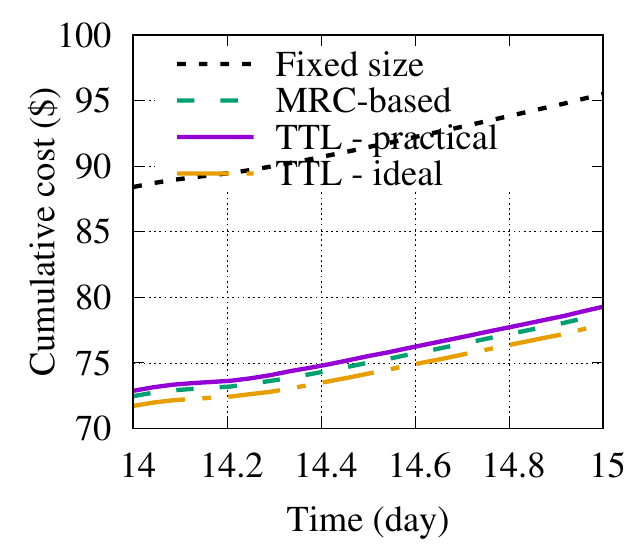}
   \caption{Cumulative cost of TTL-based approach compared to fixed-size, MRC-based, and ideal pure TTL (zoom on the right).}
   \label{fig:dyn_cumulcost}
\end{figure}

In Figure\,\ref{fig:dyn_variouscost} we show the two cost components: the cumulative storage cost (left) and the cumulative misses costs (right). MRC-based solution maintains a smaller number of instances, which translates into a slighlty higher cost due to misses. Nevertheless, their sum is similar to the one obtained by the TTL-based approach, suggesting that, when we are close to the minimum, different configuration options are available.

\begin{figure}[t]
   \centering
   \hspace{-2.0em}
   \includegraphics[width=0.53\linewidth]{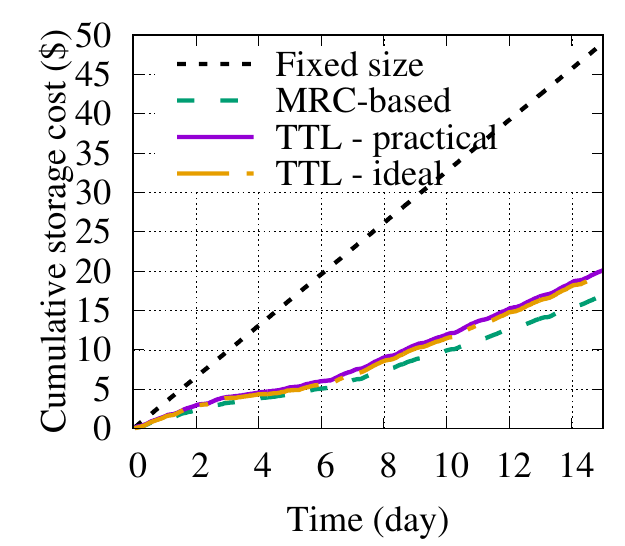}
   \hspace{-1.0em}
   \includegraphics[width=0.53\linewidth]{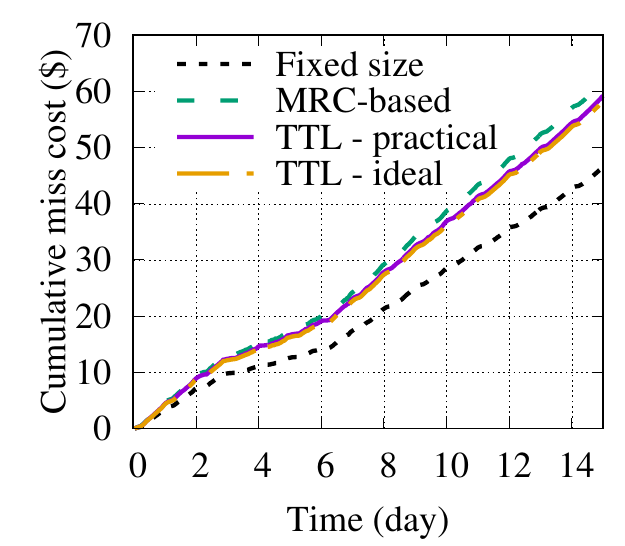}
   \caption{Cumulative storage cost (left) and cumulative miss cost (right) for different approaches.}
   \label{fig:dyn_variouscost}
\end{figure}

As we anticipated in Sect.\,\ref{subsec:complexity}, Fig.\,\ref{fig:MRCvsTTL_perf}, the computational complexity of our approach translates into a lightly loaded CPU. Instead, the MRC-based approach imposes a non negligible toll on CPU resources. This translates into additional costs, that we have not considered in this paper. In fact, as shown in   Fig.\,\ref{fig:MRCvsTTL_perf} (right) for our experimental setup, using a MRC-based approach harms the load balancer, that can ingest only half of the incoming requests. Thus, in a practical setting, the load balancer is likely to be scaled up, inducing additional fixed costs.

Figure\,\ref{fig:opt_res} compares our TTL solution with the clairvoyant TTL-OPT described in Sect\,\ref{sec:ttl-opt}. We see that there is room for even more significant cost savings: TTL-OPT achieves a cost that is one third of the baseline. TTL-OPT assumes to know the sequence of future requests and is thus unpractical. Nevertheless, this result suggests that potential improvements can come from TTL policies that use different TTL values for different contents (as TTL-OPT does) selecting the timer value on the basis of a forecast for the next inter-arrival time. In the future we plan to investigate this possibility.

\begin{figure}[t]
   \centering
   \includegraphics[width=0.53\linewidth]{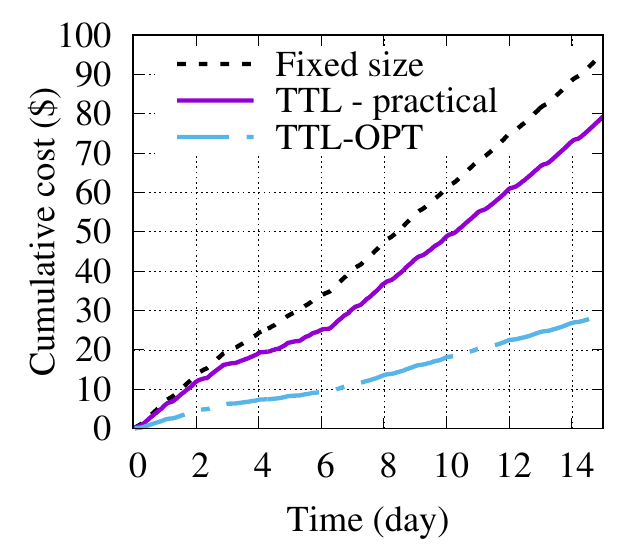}
   \caption{Optimal cost based policy: cumulative costs over time.}
   \label{fig:opt_res}
\end{figure}

Since we are dealing with a (dynamic) distributed cache, one may wonder if the assignment of object keys to cache instances is balanced. Note that Redis does not use consistent hashing, but a two-step scheme \cite{redisCluster}. There are 16384 slots, and objects keys are hashed into one of the slots. Each slot is randomly assigned to a server. When a new server is added, some randomly selected slots are transferred to the new server. When a server is removed, its slots are transferred to the other randomly selected servers.

To understand if each server maintains the responsibility of approximately the same number of slots, we have considered, for each interval, the minimum and the maximum number of assigned slots to each server, and we have normalized them with the expected number of slots per server. When there is just one server, clearly the minimum and the maximum and the expected number of slots are the same. In Figure\,\ref{fig:other_res_2} (left), we can see that each server deviated from the expected number of slots by at most 2.5\%.

\begin{figure}[htbp]
   \centering
   \includegraphics[width=0.32\linewidth]{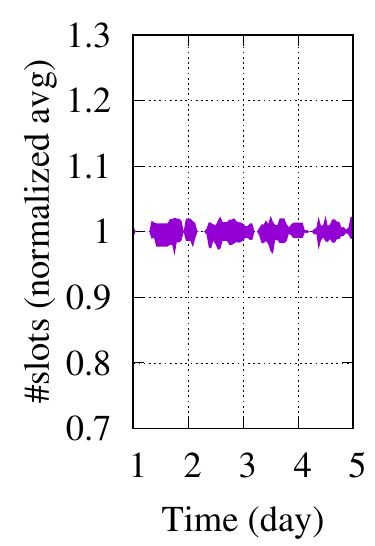}
   \includegraphics[width=0.32\linewidth]{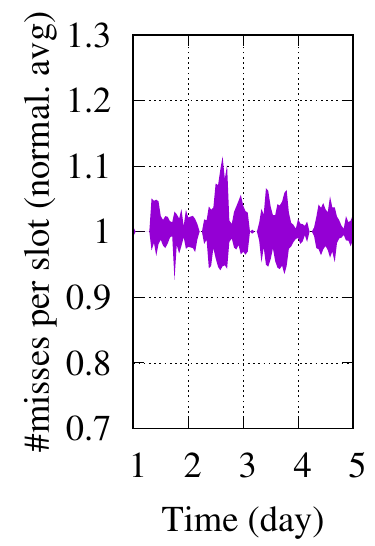}
   \includegraphics[width=0.32\linewidth]{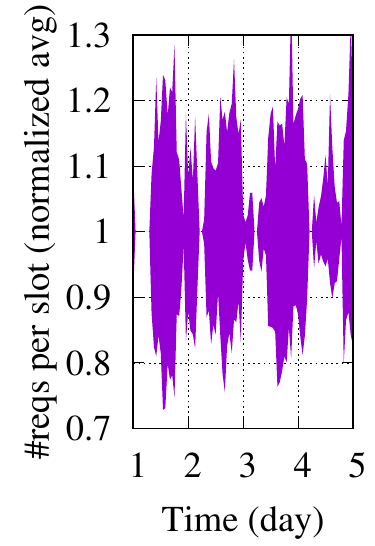}
   \caption{Normalized mean number of slots per server, misses per slot and requests per slot.}
   \label{fig:other_res_2}
\end{figure}

Similarly, we computed the number of misses per server (minimum and maximum, normalized to the total number of misses divided by the number of servers). Here the distribution among server is more spread, with servers that sometimes get 10\% more misses than the expected average. In addition, we have considered the number of requests per server. The load balancer tries to achieve distributed evenly the keys among the servers, but the actual number of requests depends on the popularity of a key.  Figure\,\ref{fig:other_res_2} (right) shows that sometimes servers need to respond to 30\% more traffic than the average. Rebalancing the server load can be done with known techniques that keeps track of the highly loaded slot and balance them among the servers, such as the ones presented in \cite{hong2013understanding} and \cite{cheng2015memory}.

We conclude our evaluation by observing that our results hold for a single-server trace collected by Akamai in a 30 days period. It is thus tempting to project our results to the current scale of a major CDN provider, for a yearly timeframe. According to the trend we measure in Fig.\,\ref{fig:dyn_cumulcost}, our approach can potentially save millions of dollars, when compared to a best practice static configuration.

\section{Conclusion}
\label{sec:concl}

Dynamic sizing of cloud caches allows cloud users to adapt the cache size to the traffic pattern and minimize their total cost, which is given by the cost of the storage and the cost of the misses. We studied a TTL-based solution to dynamically track the required cache size. We provided a theoretical lower bound for the cost achievable by TTL solutions: in fact we characterize the optimal TTL policy (TTL-OPT) when the sequence of future requests is known. Moreover, we discussed a practical low-complexity implementation of a TTL solution, and evaluated it using real-world traces. Our experiments shows that our solution obtains up to 17\% cost savings compared to a baseline static setting. 
Our results also suggests that, at least for typical CDN applications, there is no need for finer-grained billing periods or cache sizes, but most of the potential improvement is already achievable with the current offer.

Encouraged by the experimental results related to a practical TTL cache implementation, we are exploring, from a theoretical point of view, the impact of the update delay on the convergence of TTL stochastic update rule. Moreover, our comparison with TTL-OPT suggests that there are possibilities for significant additional cost savings (up to 66\%), if TTL values can be adapted on a per-content basis as a function of the specific arrival pattern.

{\footnotesize \bibliographystyle{acm}
\bibliography{cost-aware-ttl_arxiv}}

\begin{thebibliography}{10}

\bibitem{amazonElastiCache}
{Amazon Web Service ElastiCache}.
\newblock \url{https://aws.amazon.com/elasticache/}.
\newblock Accessed: Jan. 2018.

\bibitem{elastiCachePricing}
{Amazon Web Service ElastiCache Pricing}.
\newblock \url{https://aws.amazon.com/elasticache/pricing/}.
\newblock Accessed: Jan. 2018.

\bibitem{mcrouter}
{Facebook mcrouter}.
\newblock \url{https://github.com/facebook/mcrouter}.
\newblock Accessed: Jan. 2018.

\bibitem{delayCost}
{How Loading Time Affects Your Bottom Line}.
\newblock \url{https://blog.kissmetrics.com/loading-time/}.
\newblock Accessed: Jan. 2018.

\bibitem{memcached}
{Memcached}.
\newblock \url{https://memcached.org/}.
\newblock Accessed: Jan. 2018.

\bibitem{azureRedisCache}
{Microsoft Azure Redis Cache}.
\newblock \url{https://azure.microsoft.com/en-us/services/cache/}.
\newblock Accessed: Jan. 2018.

\bibitem{redis}
{Redis}.
\newblock \url{https://redis.io/}.
\newblock Accessed: Jan. 2018.

\bibitem{redisCluster}
{Redis Cluster Specification}.
\newblock \url{https://redis.io/topics/cluster-spec}.
\newblock Accessed: Jan. 2018.

\bibitem{sniaRep}
{SNIA iotta repository block I/O traces}.
\newblock \url{http://iotta. snia.org/tracetypes/3}.
\newblock Accessed: Jan. 2018.

\bibitem{Atik12}
{\sc Atikoglu, B., Xu, Y., Frachtenberg, E., Jiang, S., and Paleczny, M.}
\newblock Workload analysis of a large-scale key-value store.
\newblock In {\em Proceedings of the 12th ACM SIGMETRICS/PERFORMANCE joint
  international conference on Measurement and Modeling of Computer Systems\/}
  (2012).

\bibitem{basu2017adaptive}
{\sc Basu, S., Sundarrajan, A., Ghaderi, J., Shakkottai, S., and Sitaraman, R.}
\newblock Adaptive ttl-based caching for content delivery.
\newblock In {\em Proceedings of the 2017 ACM SIGMETRICS/International
  Conference on Measurement and Modeling of Computer Systems\/} (2017),
  pp.~45--46.

\bibitem{belady1966study}
{\sc Belady, L.~A.}
\newblock A study of replacement algorithms for a virtual-storage computer.
\newblock {\em IBM Systems journal 5}, 2 (1966), 78--101.

\bibitem{berger2017adaptsize}
{\sc Berger, D.~S., Sitaraman, R.~K., and Harchol-Balter, M.}
\newblock Adaptsize: Orchestrating the hot object memory cache in a content
  delivery network.
\newblock In {\em NSDI\/} (2017), pp.~483--498.

\bibitem{cao1997cost}
{\sc Cao, P., and Irani, S.}
\newblock Cost-aware www proxy caching algorithms.
\newblock In {\em Usenix symposium on internet technologies and systems\/}
  (1997), vol.~12, pp.~193--206.

\bibitem{carra2014memory}
{\sc Carra, D., and Michiardi, P.}
\newblock Memory partitioning in memcached: An experimental performance
  analysis.
\newblock In {\em Communications (ICC), 2014 IEEE International Conference
  on\/} (2014), IEEE, pp.~1154--1159.

\bibitem{che02}
{\sc Che, H., Tung, Y., and Wang, Z.}
\newblock {Hierarchical Web caching systems: modeling, design and experimental
  results}.
\newblock {\em Selected Areas in Communications, IEEE Journal on 20}, 7 (Sep
  2002), 1305--1314.

\bibitem{cheng2015memory}
{\sc Cheng, Y., Gupta, A., and Butt, A.~R.}
\newblock An in-memory object caching framework with adaptive load balancing.
\newblock In {\em Proceedings of the Tenth European Conference on Computer
  Systems\/} (2015), p.~4.

\bibitem{cidon2015dynacache}
{\sc Cidon, A., Eisenman, A., Alizadeh, M., and Katti, S.}
\newblock Dynacache: Dynamic cloud caching.
\newblock In {\em HotStorage\/} (2015).

\bibitem{coffman1973operating}
{\sc Coffman, E.~G., and Denning, P.~J.}
\newblock {\em Operating systems theory}, vol.~973.
\newblock Prentice-Hall Englewood Cliffs, NJ, 1973.

\bibitem{dehghan16}
{\sc Dehghan, M., Massouli\'e, L., Towsley, D., Menasche, D., and Tay, Y.~C.}
\newblock A utility optimization approach to network cache design.
\newblock In {\em IEEE INFOCOM 2016 - The 35th Annual IEEE International
  Conference on Computer Communications\/} (April 2016), pp.~1--9.

\bibitem{choungmo14}
{\sc Fofack, N.~C., Nain, P., Neglia, G., and Towsley, D.}
\newblock {Performance evaluation of hierarchical TTL-based cache networks}.
\newblock {\em Computer Networks 65\/} (2014), 212 -- 231.

\bibitem{garetto16}
{\sc Garetto, M., Leonardi, E., and Martina, V.}
\newblock A unified approach to the performance analysis of caching systems.
\newblock {\em ACM Trans. Model. Perform. Eval. Comput. Syst. 1}, 3 (May 2016),
  12:1--12:28.

\bibitem{hong2013understanding}
{\sc Hong, Y.-J., and Thottethodi, M.}
\newblock Understanding and mitigating the impact of load imbalance in the
  memory caching tier.
\newblock In {\em Proceedings of the 4th annual Symposium on Cloud Computing\/}
  (2013), ACM, p.~13.

\bibitem{hosseini2000optimal}
{\sc Hosseini-Khayat, S.}
\newblock On optimal replacement of nonuniform cache objects.
\newblock {\em IEEE Transactions on Computers 49}, 8 (2000), 769--778.

\bibitem{hu2015lama}
{\sc Hu, X., Wang, X., Li, Y., Zhou, L., Luo, Y., Ding, C., Jiang, S., and
  Wang, Z.}
\newblock Lama: Optimized locality-aware memory allocation for key-value cache.
\newblock In {\em Proceedings of USENIX Annual Technical Conference (USENIX
  ATC)\/} (2015), pp.~57--69.

\bibitem{hu2016kinetic}
{\sc Hu, X., Wang, X., Zhou, L., Luo, Y., Ding, C., and Wang, Z.}
\newblock Kinetic modeling of data eviction in cache.
\newblock In {\em USENIX Annual Technical Conference\/} (2016), pp.~351--364.

\bibitem{kushner03}
{\sc Kushner, H., and Yin, G.}
\newblock {\em Stochastic Approximation and Recursive Algorithms and
  Applications}.
\newblock Stochastic Modelling and Applied Probability. Springer New York,
  2003.

\bibitem{lee1999existence}
{\sc Lee, D., Choi, J., Kim, J.-H., Noh, S.~H., Min, S.~L., Cho, Y., and Kim,
  C.~S.}
\newblock On the existence of a spectrum of policies that subsumes the least
  recently used (lru) and least frequently used (lfu) policies.
\newblock In {\em ACM SIGMETRICS Performance Evaluation Review\/} (1999),
  vol.~27, ACM, pp.~134--143.

\bibitem{lim2010automated}
{\sc Lim, H.~C., Babu, S., and Chase, J.~S.}
\newblock Automated control for elastic storage.
\newblock In {\em Proceedings of the 7th international conference on Autonomic
  computing\/} (2010), ACM, pp.~1--10.

\bibitem{lorido2014review}
{\sc Lorido-Botran, T., Miguel-Alonso, J., and Lozano, J.~A.}
\newblock A review of auto-scaling techniques for elastic applications in cloud
  environments.
\newblock {\em Journal of Grid Computing 12}, 4 (2014), 559--592.

\bibitem{Mattson}
{\sc Mattson, R.~L., Gecsei, J., Slutz, D.~R., and Traiger, I.~L.}
\newblock Evaluation techniques for storage hierarchies.
\newblock {\em IBM Syst. J. 9}, 2 (June 1970), 78--117.

\bibitem{megiddo2003arc}
{\sc Megiddo, N., and Modha, D.~S.}
\newblock Arc: A self-tuning, low overhead replacement cache.
\newblock In {\em FAST\/} (2003), vol.~3, pp.~115--130.

\bibitem{neglia17tompecs}
{\sc Neglia, G., Carra, D., Feng, M., Janardhan, V., Michiardi, P., and
  Tsigkari, D.}
\newblock Access-time-aware cache algorithms.
\newblock {\em ACM Trans. Model. Perform. Eval. Comput. Syst. 2}, 4 (Nov.
  2017), 21:1--21:29.

\bibitem{ou2015penalty}
{\sc Ou, J., Patton, M., Moore, M.~D., Xu, Y., and Jiang, S.}
\newblock A penalty aware memory allocation scheme for key-value cache.
\newblock In {\em Proceedings of International Conference on Parallel
  Processing (ICPP)\/} (2015), pp.~530--539.

\bibitem{saemundsson2014dynamic}
{\sc Saemundsson, T., Bjornsson, H., Chockler, G., and Vigfusson, Y.}
\newblock Dynamic performance profiling of cloud caches.
\newblock In {\em Proceedings of the ACM Symposium on Cloud Computing\/}
  (2014), ACM, pp.~1--14.

\bibitem{shen2011cloudscale}
{\sc Shen, Z., Subbiah, S., Gu, X., and Wilkes, J.}
\newblock Cloudscale: elastic resource scaling for multi-tenant cloud systems.
\newblock In {\em Proceedings of the 2nd ACM Symposium on Cloud Computing\/}
  (2011), ACM, p.~5.

\bibitem{waldspurger2017cache}
{\sc Waldspurger, C., Saemundsson, T., Ahmad, I., and Park, N.}
\newblock Cache modeling and optimization using miniature simulations.
\newblock In {\em Proceedings of USENIX ATC\/} (2017), pp.~487--498.

\bibitem{waldspurger2015efficient}
{\sc Waldspurger, C.~A., Park, N., Garthwaite, A.~T., and Ahmad, I.}
\newblock Efficient mrc construction with shards.
\newblock In {\em FAST\/} (2015), pp.~95--110.

\bibitem{wang2017exploiting}
{\sc Wang, C., Urgaonkar, B., Gupta, A., Kesidis, G., and Liang, Q.}
\newblock Exploiting spot and burstable instances for improving the
  cost-efficacy of in-memory caches on the public cloud.
\newblock In {\em Proceedings of the Twelfth European Conference on Computer
  Systems\/} (2017), ACM, pp.~620--634.

\bibitem{wires2014characterizing}
{\sc Wires, J., Ingram, S., Drudi, Z., Harvey, N.~J., Warfield, A., and Data,
  C.}
\newblock Characterizing storage workloads with counter stacks.
\newblock In {\em OSDI\/} (2014), pp.~335--349.

\bibitem{xu2016blending}
{\sc Xu, Z., Stewart, C., Deng, N., and Wang, X.}
\newblock Blending on-demand and spot instances to lower costs for in-memory
  storage.
\newblock In {\em Computer Communications, IEEE INFOCOM 2016-The 35th Annual
  IEEE International Conference on\/} (2016), IEEE, pp.~1--9.

\bibitem{zhong2009program}
{\sc Zhong, Y., Shen, X., and Ding, C.}
\newblock Program locality analysis using reuse distance.
\newblock {\em ACM Trans. Program. Lang. Syst. 31}, 6 (2009), 1--39.

\end{thebibliography}

\end{document}